\documentclass[12pt,usenames,dvipsnames]{article}

\usepackage{ragged2e}
\usepackage{sectsty}

\usepackage{amsmath,amsthm,amssymb,amscd,mathtools,bm,upgreek}
\usepackage{graphicx}
\usepackage{epsfig}
\usepackage{hyperref}
\hypersetup{colorlinks,linkcolor=blue,citecolor=blue,urlcolor = blue}
\usepackage{xurl}
\usepackage{rotating}
\usepackage{float}
\usepackage{enumitem}
\usepackage{multirow}
\usepackage{booktabs}
\usepackage{eurosym}
\usepackage{caption}
\usepackage{pdfcolfoot}
\usepackage{longtable}
\usepackage[title]{appendix}
\RequirePackage{mathtools,algpseudocode,etoolbox}
\usepackage[linesnumbered,ruled,vlined]{algorithm2e}

\usepackage{listings}
\usepackage{color} 
\usepackage[style=apa,backend=biber]{biblatex}
\DeclareSourcemap{
  \maps[datatype=bibtex]{
    \map{
      \step[fieldset=issn, null]
      \step[fieldset=isbn, null]
      \step[fieldset=doi, null]
      \step[fieldset=url, null]
      \step[fieldset=note, null]
    }
  }
}
\usepackage{subcaption}
\usepackage{mwe}

\usepackage{siunitx}

\usepackage{xcolor,soul}

\usepackage{authblk}

\def\comp{\raise 1pt \hbox{$\scriptstyle\circ$}}

\def\upto{{\raise 1pt \hbox{$\scriptstyle \,\nearrow\,$}}}
\def\downto{{\raise 1pt \hbox{$\scriptstyle \,\searrow\,$}}}

\newcommand{\Halmos}{\ensuremath{\Box}}

\newtheorem{theorem}{Theorem}

\title{\Large\textbf{Secondary materials, Pigouvian taxes, and a monopsony}
}

\author[1]{Timo Kuosmanen}
\author[2,3\footnote{
Corresponding author. \newline \hspace*{5mm} 
\textit{E-mail addresses:} \texttt{timo.kuosmanen@utu.fi (T. Kuosmanen)}, \texttt{x.zhou@surrey.ac.uk (X. Zhou)}.}]{Xun Zhou}
\affil[1~]{Turku School of Economics, University of Turku, 20500 Turku, Finland}
\affil[2~]{{Surrey Business School, University of Surrey, Guildford GU2 7XH, UK}}
\affil[3~]{{Institute for Sustainability, University of Surrey, Guildford GU2 7XH, UK}}
\date{February 2025}

\begin{document}

\maketitle

\vfill
\begin{center}
Declarations of interest: none
\end{center}
\vfill

\begin{abstract}
\noindent 
Secondary materials present promising opportunities for firms to repurpose emissions into marketable goods, aligning with circular economy principles. This paper examines conditions under which introducing a market for secondary materials can completely replace Pigouvian emissions taxes. These conditions prove highly restrictive: positive Pigouvian emissions taxes remain necessary unless secondary materials prices immediately reach unrealistically high levels. We propose that the socially optimal budget-neutral policy is to subsidize secondary materials prices while taxing uncontrolled emissions. Further, we extend the analysis to a two-firm framework where a data center supplies residual heat to a district heating firm acting as a monopsony buyer. This extension explicitly models the demand for residual heat and explores how subsidies and emissions taxes align firm incentives with the social optimum in the absence of competitive markets. 
\\[5mm]
\\[2mm]
\end{abstract}
\vfill

\thispagestyle{empty}
\newpage
\setcounter{page}{1}
\setcounter{footnote}{0}
\pagenumbering{arabic}
\baselineskip 20pt
\setlength\bibitemsep{1.15\itemsep}

\section{Introduction}\label{sec:intro}
Secondary materials refer to materials that are not the primary outputs of industrial or manufacturing processes. These include byproducts, wastes, scraps, or used items that can be recovered, such as through reuse, repurposing, refurbishment, or recycling, depending on their quality \parencite{panchal_does_2021,kube_research_2018,rigamonti_step_2020}.\footnote{Other terms closely related to secondary materials include secondary raw materials, recoverable byproducts, recovered products, recyclable wastes, and commercial wastes, among others.} Secondary materials are becoming increasingly important across various industries, as demonstrated by a growing number of practical applications.\footnote{\textcite{kube_research_2018} provide a comprehensive review of trends in environmental and resource economics, noting that secondary materials were a prominent research focus during the early decades of the field (see, e.g., \cite{grace_secondary_1978,bingham_conditionally_1983,tsao_process_1971,anderson_tax_1977}). However, the interest in secondary materials has experienced a resurgence in recent years due to increasing concerns about the circular economy and sustainable resource management, advances in technologies for the recovery of secondary materials, and the emergence of new types of secondary materials (see, e.g., \cite{egger_resource_2024,broberg_burn_2022,curtis_when_2019,yamamoto_is_2022}).} One notable example involves the recovery of waste heat from data centers for district heating. Conventionally, district heating systems have relied on cogeneration plants burning fossil fuels \parencite{paiho_towards_2016,lund_role_2010}. Since electricity generation is increasingly based on decentralized renewable energy such as wind and solar power, replacing the cogeneration plants for district heating has proved a major challenge in Northern countries such as Finland. Utilizing waste heat from data centers offers a promising solution for zero-carbon district heating \parencite{wahlroos_future_2018,yuan_waste_2023}. For example, Google's data center in Hamina, Finland, recovers waste heat and is expected to cover approximately 80\% of the city's district heating needs.\footnote{Source: \href{https://www.fdca.fi/google-announced-e1-billion-expansion-of-hamina-data-center-and-new-heat-recovery-project/}{fdca.fi/google-announced-e1-billion-expansion-of-hamina-data-center-and-new-heat-recovery-project/} (retrieved on December 19, 2024).} Similarly, Finnish energy company Fortum and Microsoft have partnered to construct a data center region in Espoo and Kirkkonummi, Finland, which is expected to meet 40\% of the heating needs of the two cities.\footnote{Source: \href{https://www.fortum.com/data-centres-helsinki-region}{fortum.com/data-centres-helsinki-region} (retrieved on December 19, 2024).} 

These real-world examples illustrate how industries can simultaneously reduce their environmental impact and generate revenue by repurposing wastes, aligning with the principles of a circular economy. Building on this concept, Leary, Zunino, and Wagner (\citeyear{leary_marginal_2025}) (henceforth LZW) extend \citeauthor{mckitrick_derivation_1999}'s (\citeyear{mckitrick_derivation_1999}) analysis of a profit-maximizing firm subject to an emissions constraint by incorporating the recovery of secondary materials into the framework for estimating the marginal abatement cost (MAC) function. LZW's innovative approach accounts for the potential revenue generated from ``good" emissions (those converted into marketable products through abatement effort) while continuing to treat ``bad" emissions as negative externalities subject to regulations. LZW argue that, under certain conditions, secondary materials prices could replace environmental regulatory measures, such as Pigouvian emissions taxes.

While LZW's framework is an important step forward, they do not explicitly identify the conditions under which the secondary material price can fully replace the emissions tax. The first contribution of this paper is to fill this gap by deriving these conditions and assessing their feasibility. We show that the socially optimal emission target and the socially optimal emissions tax decrease as the secondary material price increases, but for the secondary material price to fully replace the emissions tax, it would need to reach unrealistically high levels immediately. Without such conditions, environmental policy instruments, such as emissions taxes, remain necessary to achieve socially optimal outcomes. 

The second contribution of our paper is to propose that a socially optimal budget-neutral policy combines subsidies for secondary materials prices with taxes on uncontrolled emissions. This policy dynamically adjusts subsidies and taxes to align firms' private incentives with social efficiency while maintaining budget neutrality and supporting circular economy principles.

In addition to the single-firm analysis, the third contribution of this paper is to extend to a two-firm framework where a data center supplies residual heat to a district heating firm, which operates as a local monopoly. By explicitly modeling the demand for residual heat, this extension addresses the absence of a competitive market and explores the roles of subsidies and emissions taxes in aligning the incentives of both firms with the social optimum. The proposed framework not only broadens the theoretical understanding of the supply and demand of secondary materials but also informs policy design in scenarios characterized by market imperfections and limited competition.

The rest of this paper proceeds as follows. Section \ref{sec:setup} introduces LZW's MAC model, which is refined and extended with formal proofs in Section \ref{sec:corr}. The extension to a two-firm model is presented in Section \ref{sec:two-firm}. Section \ref{sec:concl} concludes with policy implications and suggestions for future research.

\section{Profit-maximizing firm with secondary material revenue}\label{sec:setup}
A compelling example of revenue from secondary materials is methane capture from landfills. At the Altamont Landfill in California, harmful methane emissions are captured and converted into valuable liquefied natural gas (LNG) with global market demand.\footnote{Source: \href{https://altamontlandfill.wm.com/green-energy/index.jsp}{altamontlandfill.wm.com/green-energy/index.jsp} (retrieved on December 19, 2024).} This example illustrates how firms can generate revenue by recovering secondary materials.

Following LZW, our starting point is \citeauthor{mckitrick_derivation_1999}'s (\citeyear{mckitrick_derivation_1999}) analysis of a price-taking firm that maximizes profit subject to a regulatory constraint for emissions:
\begin{align}\label{eq:mckitrick}
    \max\; \pi(y, a) &= p_y y - C(w, y, a)\\\nonumber
    s.t. \qquad e & = e(y, a) \leq \lambda, 
\end{align}
where $y$ denotes the firm's output sold at the competitive price $p_y$, cost $C$ is a function of input price $w$, output $y$, and abatement effort $a$, and emissions $e$ are modeled as a function that increases with output $y$ but decreases with abatement effort $a$. Note that $a$ represents units of abatement effort (e.g., deploying workers or equipment for abatement tasks) rather than units of emissions abated. Finally, $\lambda$ is a command-and-control constraint imposed on emissions $e$.

\citeauthor{mckitrick_derivation_1999} shows that in the absence of regulation (i.e., $\lambda \to \infty$), the firm's privately optimal solution to \eqref{eq:mckitrick} is to produce positive levels of output $y$ and emissions $e$ while allocating no effort to abatement ($a = 0$). However, under strict regulatory constraints (i.e., low values of $\lambda$), the firm optimizes its production and abatement levels such that both $y$ and $a$ are positive.

LZW extend \citeauthor{mckitrick_derivation_1999}'s model to take into account the secondary material revenue opportunity for emissions traditionally considered to have no practical value or application beyond their release. They split \citeauthor{mckitrick_derivation_1999}'s emissions $e$ into two components as $e = e_g + e_b$. Bad emissions $e_b$ refer to emissions released untreated into the environment, consistent with McKitrick's original concept of emissions, which are subject to regulations such as command-and-control standards, emission taxes, or tradeable permits. In contrast, the firm's abatement effort $a$ converts at least part of the emissions $e_g = e - e_b$ into a marketable good (good emissions) for sale at price $p_{e_g}$.

LZW reformulate the firm's profit maximization in \eqref{eq:mckitrick} as:
\begin{equation}\label{eq:mckitrick-lzw}
    \max\; \pi(y, a) = p_y y + p_{e_g}[e(y) - e_b(y, a)] - C(w, y, a).
\end{equation}
Here, the firm generates revenue from both the sales of its primary output $y$ at price $p_y$ and the sales of good emissions $e_g = e(y) - e_b(y, a)$ as a secondary material at price $p_{e_g}$, while $C$ incorporates the additional costs associated with recovering the secondary material. LZW note that increasing $y$ raises overall emissions $e$, but abatement effort $a$ reduces $e_b$ and increases the marketable $e_g$. In some cases, while overall emissions $e$ may increase with $y$, good emissions $e_g$ can grow faster than bad emissions $e_b$.

LZW derive the first-order conditions for optimizing problem \eqref{eq:mckitrick-lzw}:
\begin{align}
    \frac{\partial \pi}{\partial y} &= p_y + p_{e_g} \left[\frac{\partial e}{\partial y} - \frac{\partial e_b}{\partial y}\right] - \frac{\partial C}{\partial y} = 0, \label{eq:foc1}\\[.5em]
    \frac{\partial \pi}{\partial a} &= -p_{e_g} \frac{\partial e_b}{\partial a} - \frac{\partial C}{\partial a} = 0. \label{eq:foc2}
\end{align}
Eq. \eqref{eq:foc1} suggests that a profit-maximizing firm chooses output $y$ such that the full marginal revenue (from both primary output and secondary material sales) equals the full marginal cost of production. Eq. \eqref{eq:foc2} suggests that the firm allocates abatement effort $a$ such that the marginal revenue of $a$, i.e., $-p_{e_g} \frac{\partial e_b}{\partial a}$, equals the marginal cost of $a$. LZW note that the marginal revenue of $a$ is positive because $\frac{\partial e_b}{\partial a} < 0$ as increased abatement effort reduces the portion of bad emissions $e_b$, and under the assumption of diminishing returns to abatement effort, the marginal revenue function of $a$ is downward-sloping.

To derive the MAC function when there is a {price} for good emissions, LZW reformulate Eq. \eqref{eq:mckitrick-lzw} as:
\begin{equation}\label{eq:mckitrick-lzw2}
    \max\; \pi(y, a) = p_y y + p_{e_g}[e - e_b] - C(w, y, a(e_b,y)).
\end{equation}
Differentiating Eq. \eqref{eq:mckitrick-lzw2} with respect to $e_b$, LZW obtain the firm's MAC function:
\begin{equation}\label{eq:mac-s}
    \text{MAC}_S = - \frac{\partial C}{\partial a} \frac{\partial a}{\partial e_b} - p_{e_g},
\end{equation}
where $-\frac{\partial C}{\partial a} \frac{\partial a}{\partial e_b}>0$ represents the MAC without secondary material sales (denoted as MAC$_M$ for \citeauthor{mckitrick_derivation_1999} in LZW). LZW note that $\text{MAC}_S$ is positive and downward-sloping, analogous to MAC$_M$. $\text{MAC}_S$ reflects the total opportunity cost of an additional unit of bad emissions $e_b$, which includes both the avoided abatement cost in $\text{MAC}_M$ and the forgone secondary material revenue from one unit of good emissions $e_g$. As a result, the firm has an incentive to abate $e_b$ by converting them into $e_g$ until $\text{MAC}_S = 0$. LZW emphasize that even in the absence of other incentives such as emissions taxes, the secondary material price $p_{e_g}$ alone can reduce the firm's uncontrolled emissions in a Pareto-improving and self-enforcing manner.

\section{Socially optimal policy with secondary material}\label{sec:corr}
In light of LZW's analysis, it might be tempting to view the establishment of secondary materials prices as a viable substitute for government interventions such as emissions taxes. However, it is important to note that the socially efficient emission target is dependent on the price $p_{e_g}$ of the secondary material associated with these emissions. We can formally prove the following:
\begin{theorem}\label{theo1}
    The socially efficient emission target $e_s^*$ is a decreasing function of {the price $p_{e_g}$ of the secondary material}. 
\end{theorem}
\begin{proof}
    See Appendix \ref{proof:theo1}.
\end{proof}
Theorem \ref{theo1} highlights the fact that the emission target is not fixed: as the price of the secondary material increases, it is beneficial to set more ambitious targets. LZW suggest that the secondary material price can completely replace the tax under certain conditions. However, they do not state those conditions explicitly. Building on Theorem \ref{theo1}, the following theorem shows that the conditions are rather restrictive.
\begin{theorem}\label{theo2}
    Establishing a price for the secondary material can completely replace the emissions tax if and only if $p_{e_g}$ is high enough to bring the firm's privately optimal choice of bad emissions $\hat e_b^0$ to zero.
\end{theorem}
\begin{proof}
    See Appendix \ref{proof:theo2}.
\end{proof}
While Theorem \ref{theo2} demonstrates the theoretical possibility of replacing the emissions tax with a price for the secondary material, such conditions are rarely feasible in practice. This emphasizes the need for complementary tools such as an emissions tax. The following theorem explores the relationship between the secondary material price and the socially optimal emissions tax:
\begin{theorem}\label{theo3}
    The socially optimal emissions tax $\tau^*$ is a decreasing function of {the price $p_{e_g}$ of the secondary material}.  
\end{theorem}
\begin{proof}
    See Appendix \ref{proof:theo3}.
\end{proof}
Theorem \ref{theo3} implies that, although achieving full substitution is often unrealistic, a price for the secondary material can partially substitute for the emissions tax. This reveals a policy synergy: an increase in the price of the secondary material not only encourages emission reduction but also reduces the reliance on the emissions tax. The synergy helps to alleviate the financial burden on firms and improve the political and economic feasibility of emission reduction policies.

Finally, if the main objective of the environmental tax is to bring emissions to an efficient level rather than collect tax revenue for the government, then it is possible to simultaneously subsidize the price of the secondary material and tax harmful emissions.
\begin{theorem}\label{theo4}
    The socially optimal budget-neutral policy is to subsidize {the price $p_{e_g}$ of the secondary material} while taxing bad emissions, such that the ratio of the socially optimal emissions tax $\tau^*$ to the socially optimal subsidy $\sigma^*$ equals the ratio of emissions abated $(e_b - e_s^*)$ to the socially efficient emission target $e_s^*$, where $e_b$ denotes the firm's initial emission level before the implementation of the subsidy and emissions tax:
    \[ \frac{\tau^*}{\sigma^*} = \frac{e_b - e_s^*}{e_s^*}. \]
\end{theorem}
\begin{proof}
    See Appendix \ref{proof:theo4}.
\end{proof}
Theorem \ref{theo4} builds on the insights from the previous theorems, integrating the relationships between the secondary material price, fiscal instruments, and socially efficient emission target into a unified framework. It reveals a dynamic relationship between the socially optimal subsidy $\sigma^*$ and the socially optimal emissions tax $\tau^*$: for a given $p_{e_g}$, as $\sigma^*$ increases, $\tau^*$ decreases, and vice versa. 

Furthermore, Theorem \ref{theo4} establishes a proportional relationship between $\sigma^*$ and $\tau^*$, linking their ratio to the ratio of emissions abated to the socially efficient emission target. This proportionality ties the fiscal instruments directly to environmental outcomes, ensuring that the subsidy and tax are dynamically adjusted based on the scale of required emission abatement.

Taken together, the dynamic and proportional relationships align the secondary material price with a carefully calibrated combination of subsidy and tax, offering a robust and flexible approach to achieving the social optimum while maintaining budget neutrality. 

\section{Monopsony buyer of secondary material}\label{sec:two-firm}
The single-firm model, as refined in the previous section, applies well to cases where competitive markets exist for secondary materials, such as the California waste facility producing LNG. In the case of residual heat from data centers, however, the buyer (typically a district heat supplier) operates as a natural monopoly. Similar to electricity and gas distribution networks, it is prohibitively expensive to build competing heating networks in the same area. In this section, we extend the single-firm model to a two-firm framework that explicitly captures the interaction between a waste heat supplier (e.g., a data center) and a monopsony buyer \parencite{robinson_economics_1969} (e.g., a district heating firm).

\subsection{Data center's problem}
Analogous to the previous sections, we assume that the data center produces primary output $y$ and converts a portion of its waste heat $e_g = e(y) - e_b(y,a)$ to supply the district heating firm. In contrast, we now assume that any unconverted waste heat is simply released without causing external harm or benefit, so there is no need to regulate or tax it. However, a subsidy $\sigma$ on the price $p_{e_g}$ of the secondary material is introduced to incentivize the recovery of waste heat. This subsidy could be provided by the government or arise from internal initiatives, such as corporate social responsibility (CSR) programs. The data center's profit maximization problem can be formulated as:
\begin{equation}
    \max\; \pi(y, a) = p_y y + [p_{e_g}+\sigma] e_g(y,a) - C(w,y,a).
\end{equation}

The first-order conditions for profit maximization are:
\begin{align}
    \frac{\partial \pi}{\partial y} &= p_y + [p_{e_g}+\sigma] \frac{\partial e_g}{\partial y} - \frac{\partial C}{\partial y} = 0, \label{eq:foc_y} \\[.5em]
    \frac{\partial \pi}{\partial a} &= [p_{e_g}+\sigma] \frac{\partial e_g}{\partial a} - \frac{\partial C}{\partial a} = 0. \label{eq:foc_g}
\end{align}
Note that dividing $\frac{\partial C}{\partial a}$ by $\frac{\partial e_g}{\partial a}$ gives the marginal cost of recovering an additional unit of waste heat ($\text{MC}_{e_g}$). Thus, we can rewrite Eq. \eqref{eq:foc_g} as:
\begin{equation}\label{eq:foc_g_2}
    p_{e_g} + \sigma = \frac{\partial C / \partial a}{\partial e_g / \partial a} = \text{MC}_{e_g}.
\end{equation}

Aside from the inclusion of $\sigma$, the only difference to the LZW analysis is that, in this case, a market for $e_g$ does not exist: there is only one potential buyer, the district heat firm. Therefore, it becomes essential to explicitly model the demand for $e_g$.  

\subsection{District heating firm's problem}
Assume the district heating firm is a local monopoly facing a downward-sloping inverse demand for heating. The firm can either buy heat $e_g$ from the data center or produce heat $h$ itself, generating emissions $e_h$ that are subject to emissions tax $\tau$. The district heating firm's profit maximization problem is:
\begin{equation}
    \max\; \pi(e_h,e_g) = p_h (h(e_h) + e_g) - p_{e_g} e_g - C(h, e_h) - \tau e_h,
\end{equation}
where $p_h$ is the price of heating sold to end-users, which is a downward-sloping function of the total heat supplied $h(e_h) + e_g$, $C(h, e_h)$ is the cost of own heating production, and $\tau e_h$ is the emissions tax paid by the firm.

The first-order conditions for profit maximization are
\begin{align}
   \frac{\partial \pi}{\partial e_h} &= \frac{\partial p_h}{\partial h} \frac{\partial h}{\partial e_h} + p_h \frac{\partial h}{\partial e_h} - \frac{\partial C}{\partial e_h} - \tau = 0, \label{eq:foc_e_h} \\[.5em]
   \frac{\partial \pi}{\partial e_g} &= \frac{\partial p_h}{\partial e_g} + p_h - p_{e_g} = 0. \label{eq:foc_e_g}
\end{align}

Eq. \eqref{eq:foc_e_h} states that a profit-maximizing district heating firm chooses the level of its own heat production such that the full marginal revenue MR$_h$, i.e., the direct marginal revenue from producing additional heat internally ($p_h \frac{\partial h}{\partial e_h}$) and the indirect price effect ($\frac{\partial p_h}{\partial h} \frac{\partial h}{\partial e_h}$), equals the total marginal cost, i.e., the marginal cost of producing additional heat through emissions $e_h$ ($\text{MC}_h=\frac{\partial C}{\partial e_h}$) and the emissions tax $\tau$. Hence, Eq. \eqref{eq:foc_e_h} can be rewritten as 
\begin{equation}\label{eq:MR1}
    \text{MR}_h=\text{MC}_h+\tau.
\end{equation}

Eq. \eqref{eq:foc_e_g} states that the district heating firm purchases heat from the data center until the full marginal revenue from selling that heat MR$_{e_g}$, i.e., the direct marginal revenue from selling additional purchased heat ($p_h$) and the indirect price effect ($\frac{\partial p_h}{\partial e_g}$), equals the marginal cost of purchasing it ($p_{e_g}$). Hence, Eq. \eqref{eq:foc_e_g} can be rewritten as 
\begin{equation}\label{eq:MR2}
    \text{MR}_{e_g}=p_{e_g}.
\end{equation}

Further, since heating from $h$ and $e_g$ are perfect substitutes, we have
\begin{equation}\label{eq:MR3}
    \text{MR}_h = \text{MR}_{e_g}.
\end{equation}
Taking Eqs. (\ref{eq:MR1}--\ref{eq:MR3}) together, this implies
\begin{equation}
    MC_h = p_{e_g} - \tau. \label{eq:heat_price}
\end{equation}
Therefore, for a district heating firm facing a downward-sloping inverse demand for heat, the firm scales its own heat production $h$ such that the marginal cost of $h$ equals the price $p_{e_g}$ of purchased heat $e_g$ minus the emissions tax $\tau$.

\subsection{Equilibrium and bargaining solutions}
Consider first the social optimum.
\begin{theorem}\label{theo5}
    Assuming well-behaved demand and cost functions, the social planner's optimum must satisfy:
    \[p_{e_g}=\text{MC}_{e_g}-\sigma^*=\text{MR}_h=MC_{h}+\tau^*.\]
\end{theorem}
\begin{proof}
    See Appendix \ref{proof:theo5}.
\end{proof}

In the real world, the price $p_{e_g}$ is determined through private bargaining between the data center and district heat firm, subject to incomplete information about each other's cost and demand functions. Moreover, it is not self-evident that the demand and cost functions for a secondary material are well-behaved. For these reasons, corner solutions are possible. 

If the bargaining solution satisfies:
\begin{equation}
    MC_{h}+\tau > p_{e_g} \geq MC_{e_g}-\sigma,
\end{equation}
then the district heat firm will purchase all $e_g$ that the data center can supply. In this case, the data center has an incentive to increase its efforts to recover more waste heat. However, the production of primary output $y$ sets an upper bound of how much heat will be generated.

If, on the other hand, the price $p_{e_g}$ is prohibitively high, that is,
\begin{equation}
    p_{e_g} > MC_{h}+\tau,
\end{equation}
then the district heat firm will find it more profitable to generate heat by itself rather than purchase the residual heat, and hence $e_g = 0$. This situation can occur if the marginal cost of the waste heat recovery is very high, even after accounting for the subsidy $\sigma$, or if the data center attempts to extract monopoly rents by overpricing the residual heat.   

To fully understand the dynamics between the data center and the district heating firm, it is important to evaluate the distinct roles played by the subsidy $\sigma$ and the emissions tax $\tau$. For the data center, the subsidy $\sigma$ reduces the effective marginal cost of waste heat recovery. If $\sigma$ is large enough to fully offset the marginal heat recovery cost $MC_{e_g}$, the data center may effectively ``donate" the residual heat. For the district heating firm, the emissions tax $\tau$ raises the effective marginal cost of internal heat production, incentivizing the firm to purchase $e_g$ from the data center, provided $p_{e_g}<\text{MC}_h+\tau$.

The interplay between $\sigma$ and $\tau$ is critical in aligning the incentives of the data center and the district heating firm to achieve the social planner's optimum. A low $\sigma$ may fail to incentivize sufficient waste heat recovery, while a too generous $\sigma$ risks inefficiencies, such as unnecessary recovery beyond demand. Similarly, a low $\tau$ may not sufficiently discourage self-production, whereas an excessively high tax could distort market behavior. Therefore, $\sigma$ and $\tau$ must be carefully calibrated to avoid corner solutions and ensure that the equilibrium price of residual heat satisfies the socially optimal condition stated in Theorem \ref{theo5}.

\section{Concluding remarks}\label{sec:concl}
This paper revisits the marginal abatement cost framework to evaluate the role of secondary material prices in achieving efficient emission reductions. While Leary, Zunino, and Wagner (\citeyear{leary_marginal_2025}) highlight the potential of {establishing prices for secondary materials} to align economic incentives with environmental goals, our analysis demonstrates their limitations as standalone mechanisms. Specifically, we show that the socially efficient emission targets decline as secondary materials prices increase, and that secondary materials prices alone cannot achieve optimal outcomes without additional policy measures.

A key contribution of our paper is that we demonstrate the necessity of positive Pigouvian emissions taxes unless secondary materials prices are immediately at such a high level that encourages firms to abate the bad emissions completely. This seems rather unrealistic considering the real-world cases: for example, Google donates for free the heat generated by its data center in Hamina, Finland, for the district heating. 

Another key contribution of this paper is that we propose a socially optimal budget-neutral policy that combines subsidies for secondary materials prices with taxes on uncontrolled emissions. By ensuring subsidies and taxes are proportionally adjusted to the scale of emission abatement required, the policy provides a robust and flexible pathway to achieving social efficiency while maintaining budget neutrality. It also aligns with circular economy principles by encouraging the productive reuse of emissions. 

The third contribution of this paper is that we extend the analysis to a monopsony buyer of the secondary material where a data center supplies residual heat to a district heating firm operating as a local monopoly. By explicitly modeling demand for residual heat, the two-firm extension demonstrates how subsidies and emissions taxes can align firm incentives with the social optimum in the naturally monopolistic setting. This extension provides actionable insights for addressing market imperfections and designing effective environmental policy instruments.

This paper emphasizes the importance of integrating market-based incentives with regulatory interventions to achieve both environmental and economic objectives. It opens fascinating revenues for future research. Empirical applications of our proposed framework could be conducted across industries where markets for secondary materials already exist or hold significant potential. The interplay between the development in the recovery of secondary materials and environmental regulations, particularly in non-competitive settings, is also worth further exploration. In addition, future research could incorporate our theoretical framework into the MAC estimation approach proposed by \textcite{Kuosmanen2021} to generate more practical and reliable MAC estimates.



\begin{appendices}
\baselineskip 20pt
\renewcommand{\thesubsection}{A.\arabic{subsection}}
\setcounter{table}{0}
\setcounter{equation}{0}
\setcounter{theorem}{0}
\renewcommand{\theequation}{A\arabic{equation}} 
\renewcommand{\thetable}{B\arabic{table}} 

\section{Proofs}\label{app:proof}
\subsection{Proof of Theorem \ref{theo1}}\label{proof:theo1}
The socially efficient emission target $e_s^*$ is determined by equating an assumed marginal damage (MD) function and the $\text{MAC}_S$ function:
\begin{equation}
    \text{MD} = \text{MAC}_S,
\end{equation}
or equivalently,
\begin{equation}\label{eq:equality}
    \text{MD} = \text{MAC}_M - p_{e_g}.
\end{equation}
We can solve for $e_s^*$ as a function of $p_{e_g}$:
\begin{equation}
    e_s^* = f(p_{e_g}),
\end{equation}
where $f$ is an implicit function defined by the equality of MD and MAC$_S$. Note that since LZW assume a price-taking firm, $p_{e_g}$ is exogenous for the firm. 

To show that $e_s^*$ is a decreasing function of $p_{e_g}$, let us differentiate both sides of Eq. \eqref{eq:equality} with respect to $p_{e_g}$:
\begin{equation}\label{eq:diff}
    \frac{\partial \text{MD}(e_s^*)}{\partial e_s^*} \frac{\partial e_s^*}{\partial p_{e_g}} = \frac{\partial \text{MAC}_M(e_s^*)}{\partial e_s^*} \frac{\partial e_s^*}{\partial p_{e_g}} -1.
\end{equation}
Rearranging Eq. \eqref{eq:diff} leads to:
\begin{equation}
    \frac{\partial e_s^*}{\partial p_{e_g}} = -\frac{1}{\frac{\partial \text{MD}(e_s^*)}{\partial e_s^*} - \frac{\partial \text{MAC}_M(e_s^*)}{\partial e_s^*}}.
\end{equation}
As noted in LZW, the MD (MAC$_M$) function is an increasing (decreasing) function of emissions, hence we have $\frac{\partial \text{MD}(e_s^*)}{\partial e_s^*} > 0$ and $\frac{\partial \text{MAC}_M(e_s^*)}{\partial e_s^*} <0$. It follows that:
\begin{equation}
    \frac{\partial e_s^*}{\partial p_{e_g}} < 0,
\end{equation}
which implies that $e_s^*$ is a decreasing function of $p_{e_g}$. \hfill \Halmos

\subsection{Proof of Theorem \ref{theo2}}\label{proof:theo2}
As shown in Eq. \eqref{eq:equality}, the socially efficient level of bad emissions $e_s^*$ is defined by the point at which the MD function equals the MAC$_S$ function, which can be rewritten as:
\begin{equation}\label{eq:equality2}
    \text{MD}(e_s^*) = \text{MAC}_M(e_s^*) - p_{e_g}.
\end{equation}

In the absence of an emissions tax, the firm chooses its optimal level of bad emissions $\hat e_b^0$ such that $\text{MAC}_s=0$, or equivalently,
\begin{equation}\label{eq:equality3}
     \text{MAC}_M(\hat e_b^0)=p_{e_g}.
\end{equation}
Inserting Eq. \eqref{eq:equality3} to Eq. \eqref{eq:equality2} gives:
\begin{equation}
     \text{MD}(e_s^*) = \text{MAC}_M(e_s^*) - \text{MAC}_M(\hat e_b^0).
\end{equation}

Theorem \ref{theo1} establishes that $e_s^*$ is a decreasing function of $p_{e_g}$. Thus, as $p_{e_g}$ increases, $e_s^*$ will decrease correspondingly, eventually reaching zero if $p_{e_g}$ is sufficiently large. At this point, we have:
\begin{equation}
     \text{MD}(e_s^*) = \text{MAC}_M(e_s^*) - \text{MAC}_M(\hat e_b^0) = 0,
\end{equation}
which implies that $\hat e_b^0=e_s^*=0$. In this case, an emissions tax is not needed. 

However, contrary to the extreme case, the secondary material price $p_{e_g}$ is usually not sufficient to drive $e_s^*$ to zero, so we typically have $e_s^*>0$. Hence,
\begin{equation}
     \text{MD}(e_s^*) = \text{MAC}_M(e_s^*) - \text{MAC}_M(\hat e_b^0) > 0,
\end{equation}
which implies that $\hat e_b^0 > e_s^*$. Therefore, a tax $\tau^*$ is needed to complement the secondary material price to bring the firm's optimal choice of bad emissions $\hat e_b^0$ to the socially efficient level $e_s^*$:
\begin{equation}\label{eq:optimal-tax}
    \tau^* = \text{MD}(e_s^*) > 0.
\end{equation} \hfill \Halmos

\subsection{Proof of Theorem \ref{theo3}}\label{proof:theo3}
To show that the socially optimal emissions tax $\tau^*$, in the typical case where $\tau^*>0$, is a decreasing function of $p_{e_g}$, let us differentiate both sides of Eq. \eqref{eq:optimal-tax} with respect to $p_{e_g}$:
\begin{equation}\label{eq:diff2}
    \frac{\partial \tau^*}{\partial p_{e_g}} = \frac{\partial \text{MD}(e_s^*)}{\partial e_s^*} \frac{\partial e_s^*}{\partial p_{e_g}}.
\end{equation}
Since we have $\frac{\partial \text{MD}(e_s^*)}{\partial e_s^*} > 0$ and $\frac{\partial e_s^*}{\partial p_{e_g}} <0$, it follows that:
\begin{equation}
    \frac{\partial \tau^*}{\partial p_{e_g}} < 0,
\end{equation}
which implies that $\tau^*$ is a decreasing function of $p_{e_g}$. \hfill \Halmos

\subsection{Proof of Theorem \ref{theo4}}\label{proof:theo4}
When the government subsidizes the secondary material price $p_{e_g}$, the socially efficient emission target $e_s^*$ is determined by:
\begin{equation}
    \text{MD}(e_s^*) = \text{MAC}_M(e_s^*) - (p_{e_g}+\sigma^*),
\end{equation}
where $\sigma^*>0$ represents the socially optimal subsidy on the secondary material price $p_{e_g}$. That is, the firm receives an effective price of $p_{e_g}+\sigma^*$ for each unit of secondary material sale. In the absence of an emissions tax, the firm chooses its optimal level of bad emissions $\hat e_b^0$ such that:
\begin{equation}
     \text{MAC}_M(\hat e_b^0)=p_{e_g} + \sigma^*.
\end{equation}

Following Theorem \ref{theo2}, unless the subsidized price $p_{e_g}+\sigma^*$ is sufficiently large to drive the socially efficient emission target $e_s^*$ to zero (a scenario that is practically unrealistic), a positive tax $\tau^* > 0$ is always needed to complement the subsidized price to bring the firm's privately optimal choice of bad emissions $\hat e_b^0$ to $e_s^*$. Thus, under the socially optimal condition, both $\tau^*$ and $\sigma^*$ are typically positive.

To ensure that the socially optimal policy is budget-neutral, the revenue generated from the emissions tax must equal the cost of the subsidy:
\begin{equation}\label{eq:budget-neutral}
    \tau^* \cdot e_s^* = \sigma^* \cdot (e_b - e_s^*),
\end{equation}
where $e_b$ represents the firm's initial emission level before the implementation of the subsidy and emissions tax. Rearranging Eq. \eqref{eq:budget-neutral} yields the following relationship:
\begin{equation}\label{eq:budget-neutral2}
    \frac{\tau^*}{\sigma^*} = \frac{e_b - e_s^*}{e_s^*},
\end{equation}
which implies that budget neutrality can be achieved if $\tau^*$ and $\sigma^*$ are proportional to the amount of emissions abated relative to the socially efficient emission target $e_s^*$.

Therefore, by setting $\tau^* > 0$ and $\sigma^* > 0$ to satisfy Eq. \eqref{eq:budget-neutral2}, the government achieves the socially optimal emission level $e_s^*$, where the subsidy $\sigma^*$ for the secondary material price $p_{e_g}$ and the tax $\tau^*$ on bad emissions ensure budget neutrality with no net revenue or cost for the government. \hfill \Halmos

\subsection{Proof of Theorem \ref{theo5}}\label{proof:theo5}
Combining the first-order conditions (and their rewritten forms) of the two firms, i.e., Eqs. (\ref{eq:foc_y}--\ref{eq:foc_g_2}) for the data center and Eqs. (\ref{eq:foc_e_h}--\ref{eq:MR2}) for the district heating firm, as well as the relationship established in Eqs. (\ref{eq:MR3}--\ref{eq:heat_price}), we can derive the social optimum condition as follows:
\begin{equation}\label{eq:social-optimum}
    p_{e_g}=\text{MC}_{e_g}-\sigma^*=\text{MR}_h=MC_{h}+\tau^*.
\end{equation}
\hfill \Halmos

\end{appendices}


\printbibliography


\end{document}